\documentclass[12pt,reqno]{article}

\usepackage{verbatim}
\usepackage[usenames]{color}
\usepackage{amssymb}
\usepackage{graphicx}
\usepackage{amscd}

\usepackage[colorlinks=true,
linkcolor=webgreen,
filecolor=webbrown,
citecolor=webgreen]{hyperref}

\definecolor{webgreen}{rgb}{0,.5,0}
\definecolor{webbrown}{rgb}{.6,0,0}

\usepackage{color}
\usepackage{fullpage}
\usepackage{float}

\usepackage{graphics,amsmath,amssymb}
\usepackage{amsthm}
\usepackage{amsfonts}
\usepackage{latexsym}
\usepackage{epsf}

\DeclareMathOperator{\nuc}{nuc}
\DeclareMathOperator{\mnuc}{mnuc}
\DeclareMathOperator{\crep}{crep}
\DeclareMathOperator{\facge}{facge2}
\DeclareMathOperator{\circsf}{circsf}
\DeclareMathOperator{\sqoh}{sq021}
\DeclareMathOperator{\sqfoh}{sqfree021}
\DeclareMathOperator{\sqtwo}{sq2120}
\DeclareMathOperator{\sqftwo}{sqfree2120}
\DeclareMathOperator{\currie}{currie}
\DeclareMathOperator{\testone}{test021}
\DeclareMathOperator{\testtwo}{test2120}
\DeclareMathOperator{\isBorder}{isBorder}
\DeclareMathOperator{\isBorderCone}{isBorderC1}
\DeclareMathOperator{\isBorderCtwo}{isBorderC2}
\DeclareMathOperator{\isBorderCthree}{isBorderC3}
\DeclareMathOperator{\isBordered}{isBordered}
\DeclareMathOperator{\isAlternatingzero}{isAlternating0}
\DeclareMathOperator{\hasMNUCO}{hasMNUCO}
\DeclareMathOperator{\hasMNUCE}{hasMNUCE}

\DeclareMathOperator{\isAlternatingE}{isAlternatingE}

\usepackage{breakurl}

\begin{document}

\theoremstyle{plain}
\newtheorem{theorem}{Theorem}
\newtheorem{corollary}[theorem]{Corollary}
\newtheorem{lemma}[theorem]{Lemma}
\newtheorem{proposition}[theorem]{Proposition}

\theoremstyle{definition}
\newtheorem{definition}[theorem]{Definition}
\newtheorem{example}[theorem]{Example}
\newtheorem{conjecture}[theorem]{Conjecture}

\theoremstyle{remark}
\newtheorem{remark}[theorem]{Remark}

\author{
Trevor Clokie, Daniel Gabric, and 
Jeffrey Shallit\\
School of Computer Science\\
University of Waterloo\\
Waterloo, Ontario N2L 3G1\\
Canada\\
\href{mailto:trevor.clokie@uwaterloo.ca}{\tt trevor.clokie@uwaterloo.ca}\\ 
\href{mailto:dgabric@uwaterloo.ca}{\tt dgabric@uwaterloo.ca}\\
\href{mailto:shallit@uwaterloo.ca}{\tt shallit@uwaterloo.ca}
}

\title{Circularly squarefree words and unbordered conjugates:  a new approach}

\maketitle

\begin{abstract}
Using a new approach based on automatic sequences, logic, and a decision procedure,
we reprove some old theorems about circularly squarefree words and unbordered conjugates in a new and simpler way.
Furthermore, we prove three new results about unbordered conjugates:  we complete the classification, due to Harju and Nowotka, of binary words with the maximum number of unbordered conjugates; we prove that for every possible number, up to the maximum, there exists a word having that number of unbordered conjugates, and finally, we determine the expected number of unbordered conjugates in a random word.
\end{abstract}

\section{Introduction}
Throughout this paper, $\Sigma_k$ denotes
the alphabet $\{ 0, 1, \ldots, k-1 \}$.

Two words are said to be {\it conjugate\/} if one is a cyclic shift of the other, as in the English words {\tt enlist} and {\tt listen}.

A word $w$ has a {\it border} $x$ if $x \not\in \lbrace \epsilon, w\rbrace$ and $x$ is both a prefix and suffix of $w$; the two occurrences
of $x$ are allowed to overlap each other.
For example,
{\tt alfa} is a border of {\tt alfalfa}.
A word $w$ is said to be {\it bordered\/} if it has a border, and otherwise, it is {\it unbordered}.  It follows immediately from the Lyndon-Sch\"utzenberger theorem 
\cite{Lyndon&Schutzenberger:1962} that a word
$w$ if bordered iff it has a border of length
$\leq |w|/2$; then the two shorter borders cannot overlap each other.  For example, {\tt alfalfa} is also bordered by {\tt a}.

A word $w$ is said to be a {\it square\/} if $w = xx$ for
some nonempty word $x$.  An example in French is the word {\tt couscous}.   A word is 
{\it squarefree\/} if no nonempty factor is a square. 
Let $\mu$ be the {\it Thue-Morse morphism}, defined by
$\mu(0) = 01$ and $\mu(1) = 10$.   The Thue-Morse word
$\bf t = {\tt 01101001} \cdots$ is the fixed point, starting with $0$, of $\mu$.   Thue \cite{Thue:1906,Thue:1912,Berstel:1995} proved that there exist infinite squarefree words over a three-letter alphabet; also see \cite{Allouche&Shallit:1999}.   A famous example of such a word can be obtained from the Thue-Morse word as follows:  count the number of $1$'s between two consecutive $0$'s in $\bf t$.  This gives the so-called {\it ternary Thue-Morse word} 
$$ {\bf c} = 210201 \cdots,$$
and is squarefree.   An alternative description of 
$\bf c$ is as follows:  it is the image, under $\tau$
of the fixed point of the morphism $\varphi$ defined
below:
\begin{align*}
\varphi(0) &= 01 \quad & \tau(0) &= 2 \\
\varphi(1) &= 20 \quad & \tau(1) &= 1 \\
\varphi(2) &= 23 \quad &\tau(2) &= 0 \\
\varphi(3) &= 02 \quad & \tau(3) &= 1
\end{align*}

A word $w$ is {\it circularly squarefree\/} if every one of its conjugates is squarefree.
For example, {\tt outshout} is squarefree, but not circularly squarefree.  Clearly we have

\begin{proposition}
A word is circularly squarefree iff all its conjugates are unbordered.
\end{proposition}

We now turn to a description of what we do in this paper.
Using a complicated case-based argument, Currie \cite{Currie:2002b} proved that there exist circularly squarefree ternary words of every length $n$, except for $\{ 5, 7, 9, 10, 14, 17 \}$.   The first of our main results is a new proof of Currie's theorem, based on the following result:

\begin{theorem}
For all natural numbers $n > 3$, except 
$5, 7, 9, 10, 14, 17, 21$, and $28$, there exists
a factor $x = x(n)$ of the ternary Thue-Morse word $\bf c$
that is either
\begin{itemize}
    \item[(a)] of length $n-3$, and $x021$ is circularly squarefree;
    \item[(b)] of length $n-4$, and $x2120$ is circularly squarefree.
\end{itemize}
\label{currie}
\end{theorem}

We now turn to unbordered conjugates.
In two fundamental papers, Harju and Nowotka \cite{Harju&Nowotka:2004,Harju&Nowotka:2008} studied the unbordered conjugates of a word.
In particular, letting $\nuc(w)$ denote the number of unbordered conjugates of $w$,
and $\mnuc_k (n)$ denote the maximum number of unbordered conjugates of a length-$n$ word over a $k$-letter alphabet, they proved that
\begin{itemize}
    \item[(a)] 
for binary words $w$ of length $n \geq 4$ we have $\nuc(w) \leq n/2$;

\item[(b)] for $n>2$ even, there exists a binary word of length $n$ having $n/2$ unbordered conjugates iff 
$n = 2^k$ or $n = 3 \cdot 2^k$ for some $k \geq 1$.

\end{itemize}

In other words, they explicitly computed
$\mnuc_2 (n)$ for all even $n$ and bounded it above for odd $n$.  We complete the understanding of $\mnuc_2(n)$ by proving that $\mnuc_2(n) = \lfloor n/2 \rfloor$ for all
odd $n > 3$.
Our strategy is to show that the maximum of $\nuc(w)$, over all words of length $n$, is actually achieved by a factor of the Thue-Morse word.   

More precisely, we prove
\begin{theorem}\label{theorem:maxUn}
For all $n \geq 1$, there exists a length-$n$ factor $w$ of the Thue-Morse word $\bf t$ with
$\nuc(w) = \mnuc_2 (n)$.  Furthermore, such a factor is guaranteed to occur starting at a position $\leq n$ in $\bf t$.
\end{theorem}

\section{Circularly squarefree ternary words via {\tt Walnut}}

Since the ternary Thue-Morse word $\bf c$ is squarefree, it is reasonable to hope its factors might be a good source of
circularly squarefree words.    Unfortunately, $\bf c$ contains
circularly squarefree words of length $n$ for only about
$1/8$ of all natural numbers $n$, as the following result
shows.

\begin{theorem}
There is a length-$n$ factor of $\bf c$ that is circularly squarefree iff $(n)_2$ is accepted by the automaton in Figure~\ref{fig1}.
\begin{figure}[H]
    \centering
    \includegraphics[width=6.5in]{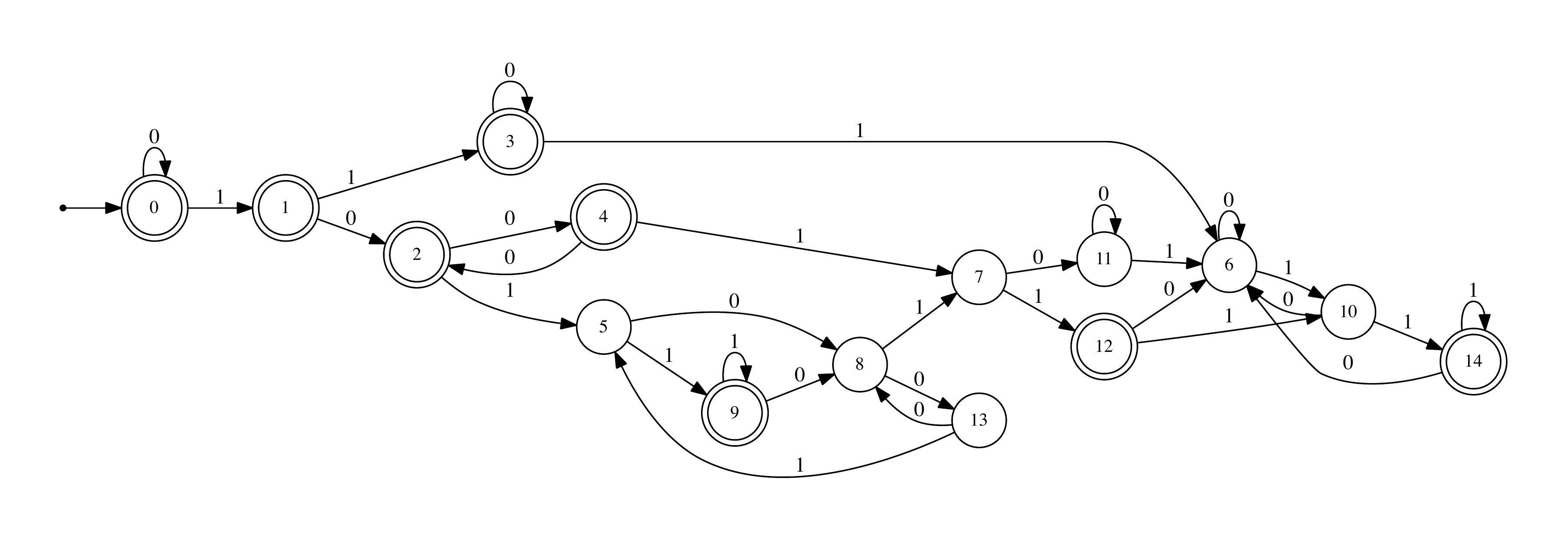}
    \caption{Automaton accepting lengths $(n)_2$ of circularly squarefree words occurring in $\bf c$}
    \label{fig1}
\end{figure}

\end{theorem}

To prove this result, we make use of the fact that many first-order
statements concerning claims about $k$-automatic sequences are decidable \cite{Bruyere&Hansel&Michaux&Villemaire:1994}.   Furthermore, there is free software called {\tt Walnut} available
to decide these claims \cite{Mousavi:2016}.

Let $(n)_k$ denote the canonical base-$k$
representation of $n$, starting with the most significant
digit, having no leading zeroes.  A sequence $(a_n)_{n \geq 0}$
is {\it $k$-automatic\/} if there is a deterministic finite automaton
with output (DFAO) taking $(n)_k$ as input, and reaching a state
with $a_n$ as output.   For example, Figure~\ref{fig0} illustrates
the DFAO generating the sequence $\bf c$.  The notation
$q/a$ in a state means the name of the state is $q$ and the output is $a$.
\begin{figure}[H]
    \centering
    \includegraphics[width=4in]{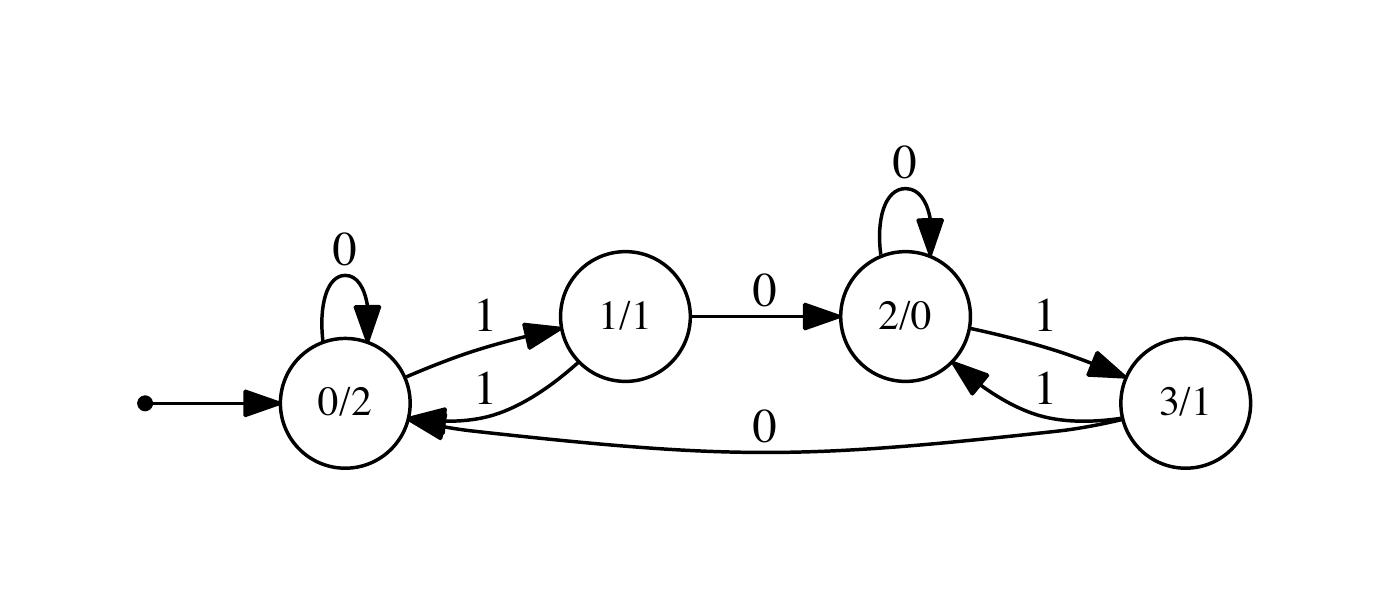}
    \caption{DFAO computing the sequence $\bf c$}
    \label{fig0}
\end{figure}
For more about automatic sequences, see \cite{Allouche&Shallit:2003}.

\begin{proof}
We can use the ideas in \cite{Shallit&Zarifi:2019}, adapted for our case.  We create first-order logical predicates $\crep$, $\facge$, and $\circsf$ 
as follows:
\begin{itemize}
    \item $\crep(i,m,p,n,s)$ evaluates to
{\tt true} iff in the length-$n$ word (considered circularly) starting
at position $s$ of the word $\bf c$, there is a factor $w$ of length $m$ and (not necessarily least) period $p \geq 1$ starting at position $i$;
    \item $\facge(n,s)$ evaluates to {\tt true} iff in the 
    length-$n$ word (considered circularly) starting at position $s$ of the word $\bf c$ there is a square or higher power;
    \item $\circsf(n)$ evaluates to {\tt true} iff some length-$n$ factor (considered circularly) of the word $\bf c$ has no squares.
\end{itemize}
\begin{small}
\begin{align*}
\crep(i,m,n,p,s) & := \exists j \,  ((j\geq i)\land (j+p<s+n)\land (j+p<i+m)) \implies C[j]=C[j+p]) \land \\
& (\forall j  \, ((j\geq i)\land (j<s+n)\land (j+p\geq s+n)\land(j+p<i+m)) \implies C[j]=C[j+p-n]) \land \\
&(\forall  j \,  ((j\geq i)\land (j\geq s+n) \land (j+p<i+m)) \implies C[j-n]=C[j+p-n])\\
\facge(n,s) & :=  \exists i,m,p \ (p\geq 1) \land (m \leq n) \land (i\geq s) \land (i<s+n) \land (m\geq 2p) \land \crep(i,m,n,p,s)\\
\circsf(n) & := \exists s \, \neg \facge(n,s)
\end{align*}
\end{small}
When we evaluate these predicates in {\tt Walnut}, we get the
automaton depicted in Figure~\ref{fig1}.  It accepts those $(n)_2$ for which {\tt circsf} evaluates to {\tt true}.
\end{proof}

\begin{remark}
All the {\tt Walnut} code for the theorems in this paper is available at\\
\centerline{\url{https://cs.uwaterloo.ca/~shallit/papers.html} \ .}
The reader can therefore verify our results.
\end{remark}

\begin{corollary}  The number of lengths $\ell$, with
$2^n \leq \ell < 2^{n+1}$ and $n \geq 4$,
such that $\bf c$ contains a
circularly squarefree factor of length $\ell$, is
$2^{n-3} - F_{n-3} + 2$, where $F_n$ is the $n$'th Fibonacci
number.
\end{corollary}

\begin{proof}
By standard techniques, by determining the roots of the characteristic polynomial of the $15 \times 15$ matrix encoding
transitions of the automaton in Fig.~\ref{fig1}.
\end{proof}

So while the factors of the ternary Thue-Morse word alone do not suffice for our purpose,
it turns out that a small modification of them do.
We now give the proof of our first main result, Theorem~\ref{currie}.

\begin{proof} (of Theorem~\ref{currie})
Let $n \ge 4$ and $w \in \{ x021, y2120 \}$, where $x,y$ are factors of the ternary Thue-Morse word $\bf c$ of lengths $n-3$ and $n-4$, respectively. 

First, we create a predicate $\sqoh(i,n,p,s)$ which evaluates to {\tt true} if $w' := x021x02$ contains a square of order $p$ with $p \geq 1$ and $2p \le n$ beginning at index $i-s$, where 
$x = {\bf c}[s..s+n-4]$. We do this by defining $w[j]$ for all $j$ such that $i\le j < i+p$ as follows:
$$w[j] = \begin{cases}
    {\bf c}[j], & \text{if } j < s+n-3; \\
         0,         &\text{if } j \in \{s+n-3,\ s+2n-3\}; \\
         2,         &\text{if } j \in \{s+n-2,\ s+2n-2\}; \\
         1,         &\text{if } j = s+n-1; \\
    {\bf c}[j-n], & \text{if } s+n \le j < s+2n-3.
    \end{cases}$$
The goal is that $\sqoh$ should represent
the implication
$$\forall j\, ((i\le j)
\land (j< i+p)) \implies w[j] = w[j+p].$$ 
It is formed by constructing the conjunction of the predicates 
$$\forall j\,  ((i\le j)
\land (j< i+p) \land (w[j] = \alpha) \land (w[j+p] = \beta)) \implies \alpha = \beta$$
for each possible combination $j$ and $j+p$, and simplifying.   

Next, we create a second predicate $\sqfoh(n,s)$, which evaluates to {\tt true} if there exists $x$ where $w = x021$ is circularly squarefree, for the given values of $n$ and $s$:
\begin{multline*}
\sqfoh(i,n,p,s) := 
(n > 3) \land (\forall i,p \, ((1 \leq p) \land (2p \leq n) \land (s \leq i) \land (i < s+n)) \\
    \implies \neg(\sqoh(i,n,p,s))) .
\end{multline*}
Similarly, we create the analogous predicates  $\sqtwo(i,n,p,s)$ and $\sqftwo(n,s)$ for the word $w' := y2120y212$.

Finally, the predicates
\begin{align*}
\testone(n) := \exists s \, \sqfoh(n,s) \\
\testtwo(n) := \exists s \, \sqftwo(n,s) 
\end{align*}
return {\tt true} if there exists a length-$n$ squarefree word formed by concatenating
some factor of $\bf c$ with $021$
(respectively, $2120$).
The automaton for $\testone(n)$ is
depicted in Figure~\ref{cfig1} and the automaton for
$\testtwo(n)$ is depicted in
Figure~\ref{cfig2}.
\begin{figure}[H]
    \centering
    \includegraphics[width=6in]{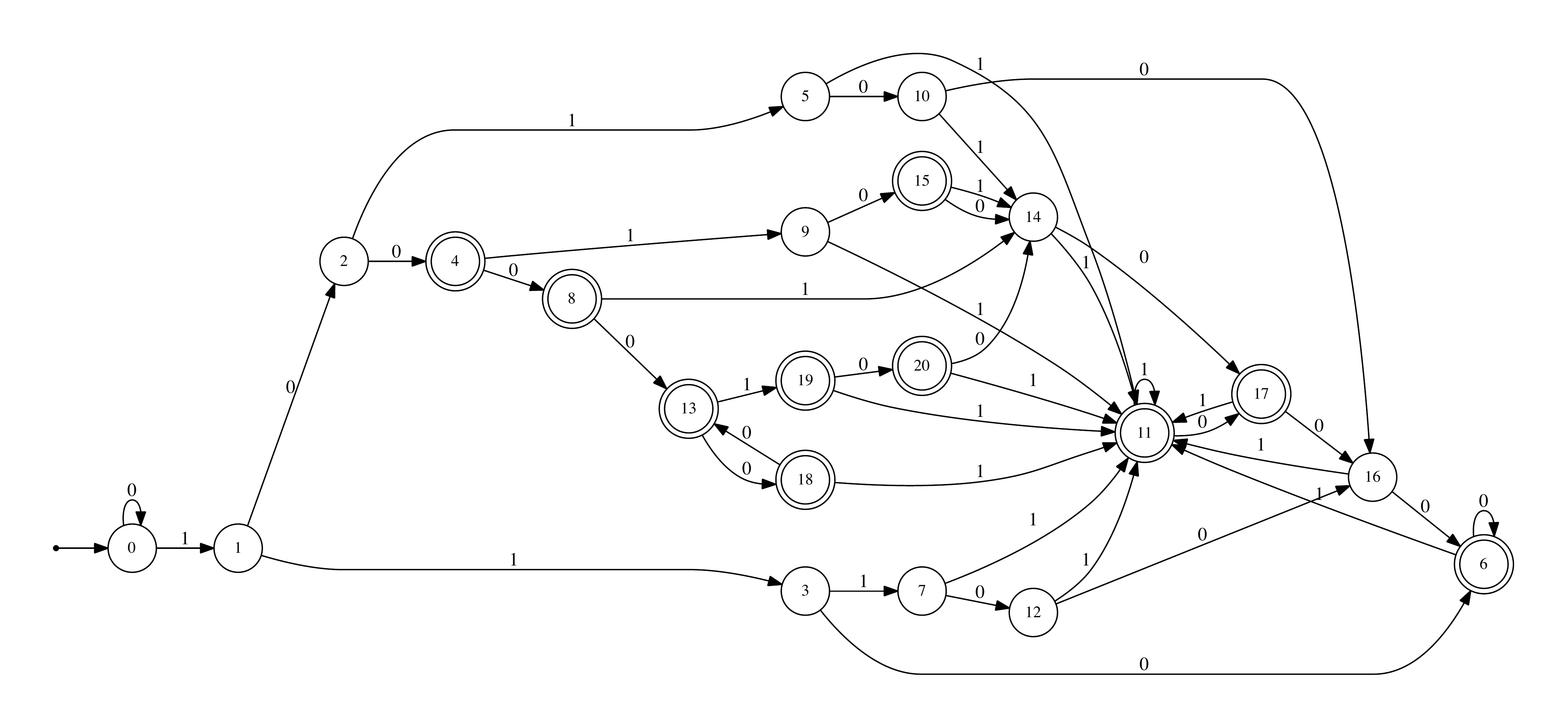}
    \caption{DFA computing $\exists s \, \sqfoh(n,s)$}
    \label{cfig1}
\end{figure}
\begin{figure}[H]
    \centering
    \includegraphics[width=6in]{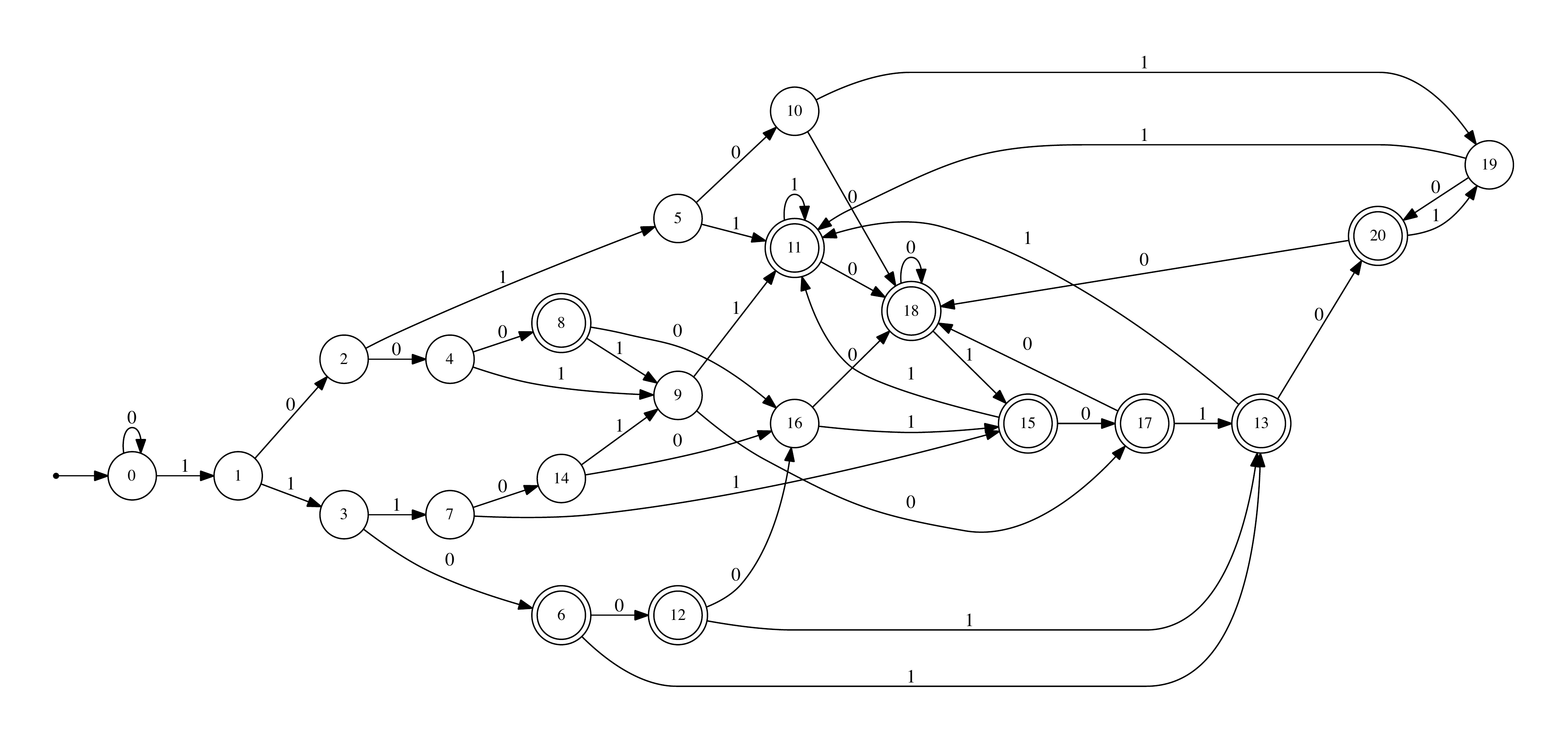}
    \caption{DFA computing $\exists s \, \sqftwo(n,s)$}
    \label{cfig2}
\end{figure}

When we now evaluate the predicate
$$ \currie(n) :=  \testone(n) \lor \testtwo(n) $$
with {\tt Walnut}, we get the automaton depicted
in Figure~\ref{curriefig}.
\begin{figure}[H]
    \centering
    \includegraphics[width=6in]{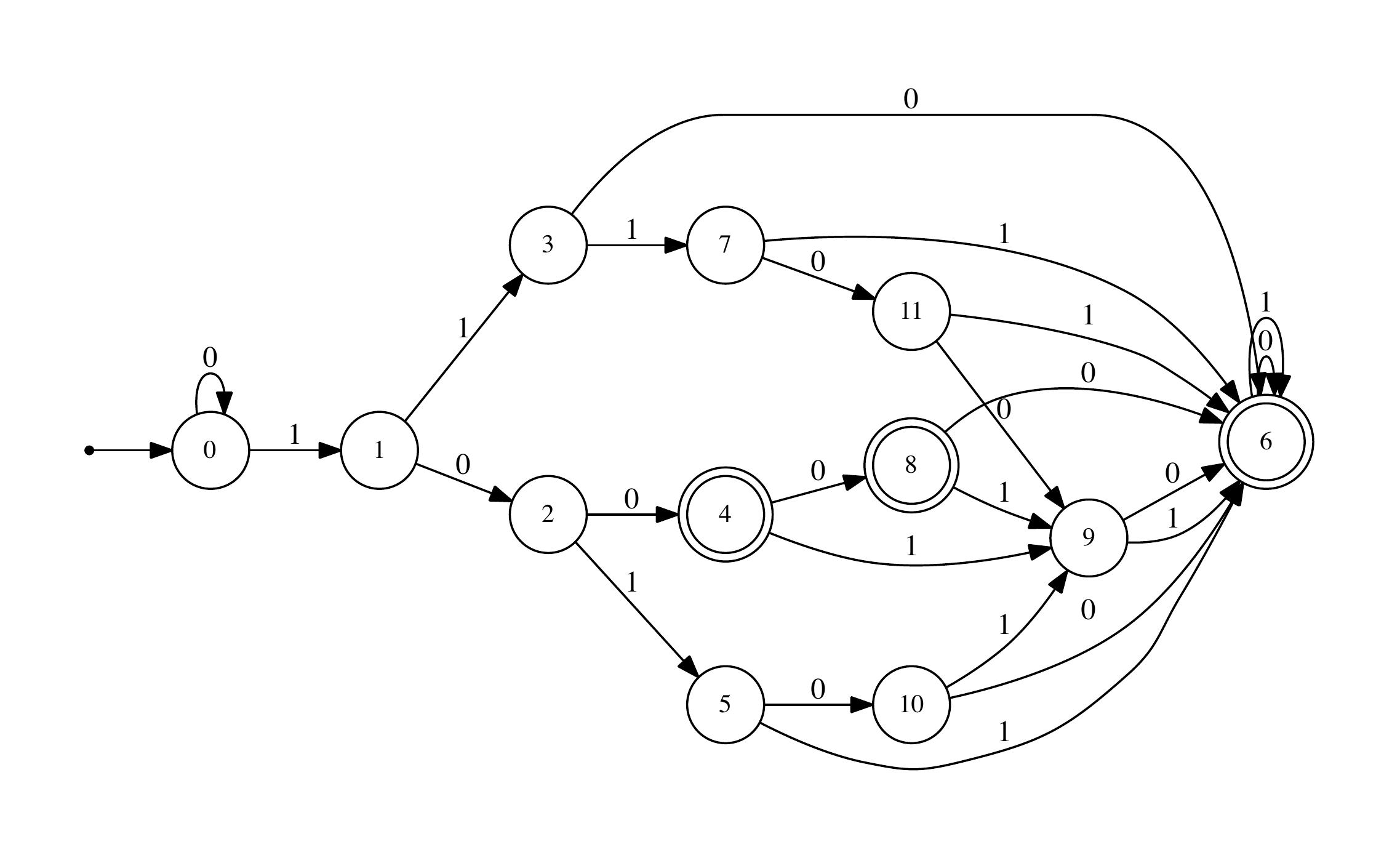}
    \caption{DFA computing acceptable $n$}
    \label{curriefig}
\end{figure}
By inspection we easily see that the automaton in Figure~\ref{curriefig}
accepts the base-$2$ representation of all $n$ except
$0,1,2,3,5,7,9,10,14,17,21,28$. 
\end{proof}
As a consequence we now get Currie's theorem:
\begin{corollary}
 There exist circularly squarefree ternary words of every length $n$, except for $n \in \{ 5, 7, 9, 10, 14, 17 \}$.
\end{corollary}
\begin{proof}
Theorem~\ref{currie} gives the result for all but finitely many $n$.
It is easy to verify by a short computation that there are cyclically squarefree words of lengths $0,1,2,3, 21, 28$, and
none for lengths $5,7,9,10,14,17$.  
\end{proof}


\begin{remark}  These calculations were done in {\tt Walnut} on a Linux machine (2 CPU --- Intel E5-2697 v3 Xeon, 256 GB of RAM).
Computing the automaton for $\sqoh$ took 115.505 seconds, and the automaton for $\sqtwo$ took 124.908 seconds.
\end{remark}


\section{Unbordered conjugates}
Let $\sigma: \Sigma_k^* \to \Sigma_k^*$ denote the \emph{cyclic shift function}, where $\sigma(\epsilon) = \epsilon$, $\sigma(cw) = wc$ for $w\in \Sigma_k^*$ and $c\in \Sigma_k$. Let $\sigma^0(w) = w$ and $\sigma^i(w) = \sigma^{i-1}(\sigma(w))$ for $i\geq 1$. 

Suppose $w$ is a binary word of length $n$.
Let $\beta: \Sigma_k^* \to \Sigma_k^*$ be the \emph{border correlation function} of a word (introduced by Harju and Nowotka \cite{Harju&Nowotka:2004}),
and defined as follows:  $\beta(w) = a_0a_1\cdots a_{n-1}$, where
\[a_i = \begin{cases}u, & \text{if }\sigma^i(w)\text{ is unbordered;}\\
b, & \text{if }\sigma^i(w)\text{ is bordered.}\end{cases}\]
For example, $\beta(0001) = ubbu$ since $0001$ is unbordered, while $0010,$ and $0100$ are both bordered, and $1000$ is unbordered. Let $u,v\in \Sigma_k^*$. We say $u$ is the {\it $i$'th cyclic shift\/} of $v$ if $\sigma^i(v) = u$.

 A result from Harju and Nowotka \cite{Harju&Nowotka:2004} shows that a binary word has no two consecutive cyclic shifts that are unbordered. This result immediately tells us that a binary word of length $n$ can have at most $\lfloor n/2 \rfloor$ unbordered conjugates. For a binary word $w$ of even length to achieve this bound, every other cyclic shift must be unbordered, or, in other words either $\beta(w) = (ub)^{|w|/2}$ or $\beta(w) = (bu)^{|w|/2}$. Harju and Nowotka \cite{Harju&Nowotka:2004} showed that the only words of even length that achieve this bound are the circularly overlap-free words, which are of length $3\cdot 2^i$ and $2^i$ for $i\geq 1$.

 Let $w$ be a binary word. Suppose $w$ is of even length and is not circularly overlap-free. Clearly $w$ cannot have $|w|/2$ unbordered conjugates, but  it could potentially have $|w|/2-1$ unbordered conjugates. Then $\beta(w) = (ub)^ib(ub)^{|w|/2-i-1}b$ for some $i\geq 0$, up to conjugation. Now suppose $w$ is of odd length. No circularly overlap-free words exist of odd length, so it makes sense to think that $w$ could contain a maximum of $\lfloor |w|/2\rfloor$ unbordered conjugates. Then $\beta(w) = (ub)^{\lfloor |w|/2\rfloor} b$, up to conjugation.

Let $w$ be a bordered binary word. Then $w=uvu$ for some words $u$ and $v$. We say $w[1..|u|]$ is the {\it first border\/} of $w$, and $w[|w|-|u|+1..|w|]$ is the {\it second border\/} of $w$.

Now we prove Theorem~\ref{theorem:maxUn}.
\begin{proof}
When $n=1,2,3$ the maxmium number of unbordered conjugates $\mnuc_2(n)$ is achieved by the words $0$, $01$, and $011$ respectively. Specifically we have that $\mnuc_2(1)=1$, $\mnuc_2(2) = 2$, and $\mnuc_2(3) = 2$. It is readily verified that each of these words occur as a factor of the Thue-Morse word at position $\leq n$.

Let $w$ be a length-$n$ word at position $m$ of the Thue-Morse word. The first step is to create a first-order predicate $\isBorder(l,m,n)$ that asserts that a cyclic shift of $w$ has a border of a certain length.  More specifically, we want to know whether the $l$'th cyclic shift of $w$ has a border of length $k$. There are three cases to consider. 

\begin{enumerate}
\item When a prefix of the second border is a suffix of $w$ and a suffix of the second border is a prefix of $w$. In other words, $w=yuvx$ for words $u,v,x,y$ where $xy = u$, $|y| = l$, and $|u|=k$. This predicate is denoted by $\isBorderCone(k,l,m,n)$.
\item When both borders are completely contained inside of $w$. In other words, $w=yuux$ for words $y,u,x$ where $|yu|=l$, and $|u|=k$. This predicate is denoted by $\isBorderCtwo(k,l,m,n)$.
\item When a prefix of the first border is a suffix of $w$ and a suffix of the first border is a prefix of $w$. In other words, $w=yvux$ for words $u,v,x,y$ where $xy = u$, $|yvu|=l$, and $|u|=k$. This predicate is denoted by $\isBorderCthree(k,l,m,n)$.
\end{enumerate}
\begin{align*}
\footnotesize
\isBorderCone(k,l,m,n) & := ((k+l > n) \Rightarrow ((\forall i (i<n-l) \Rightarrow T[m+l+i] = T[m+l-k+i]) \\
             & \land  (\forall i (i<k+l-n) \Rightarrow T[m+i] = T[m+n-k+i])))\\
\isBorderCtwo(k,l,m,n) & :=  (((k+l \leq n) \land (l \geq k)) \Rightarrow (\forall i \, (i< k) \Rightarrow \\
& T[m+l+i] = T[m+l-k+i])) \\
\isBorderCthree(k,l,m,n) &:=  (((k+l \leq n) \land (l < k )) \Rightarrow ((\forall i \, (i< k-l) \Rightarrow 
 T[m+n-k+l+i] \\ 
 & = T[m+l+i])   
\land(\forall i \, (i< l) \Rightarrow T[m+i] = T[m+k+i])))\\
 \isBorder(k,l,m,n) &:=  \isBorderCone(k,l,m,n) \land \isBorderCtwo(k,l,m,n) \land \isBorderCthree(k,l,m,n) .
\end{align*}

We define the predicate $\isBordered(l,m,n)$ that asserts that the $l$'th cyclic shift of a length $n$ word at position $m$ in the Thue-Morse word is bordered. We can create this predicate by checking whether this word has a border of size $\leq n/2$.
\[ \isBordered(l,m,n) := \exists i (2i \leq n \land i \geq 1 \land \isBorder(i,l,m,n)) .\] 
 
Recall that when $|w|$ is odd and $w$ has a maximum number of unbordered conjugates, we have that $\beta(w) = (ub)^{\lfloor |w|/2\rfloor} b$, up to conjugation. So we have exactly one pair of adjacent bordered cyclic shifts, and the rest of the cyclic shifts of $w$ alternate between bordered and unbordered. The predicate $\isAlternatingzero(l,m,n)$ asserts that all of the cyclic shifts of a length $n$ word at position $m$ in the Thue-Morse word alternate between unbordered and bordered, except for the $l$'th and $l+1$'th cyclic shifts, which are both bordered.
\begin{align*}
&\isAlternatingzero(l,m,n) := \\
&\forall i (((i \neq l \land i<n-1) \Rightarrow (\isBordered(i,m,n) = \lnot \isBordered(i+1,m,n))))\land  \\
&(((i \neq l) \land (i=n-1)) \Rightarrow (\isBordered(n-1,m,n) = \lnot \isBordered(0,m,n))) . 
\end{align*}

Now we create a predicate $\hasMNUCO(m,n)$ that asserts that a length $n$ word at position $m$ in the Thue-Morse word achieves the maximum number of unbordered conjugates.
\begin{align*}
&\hasMNUCO(m,n) := \exists i (((i<n-1 \land \isBordered(i,m,n) \land \isBordered(i+1,m,n)) \lor \\
&(i=n-1 \land \isBordered(n-1,m,n) \land \isBordered(0,m,n))) \land \nonumber 
\isAlternatingzero(i,m,n)) .
\end{align*}

Similarly, recall that when $|w|$ is even and $w$ has a maximum number of unbordered conjugates, we have that $\beta(w) = (ub)^ib(ub)^{|w|/2-i-1}b$ for some $i\geq 0$ or $\beta(w) = (ub)^{|w|/2}$, up to conjugation. So we have that either all of the cyclic shifts of $w$ alternate between bordered and unbordered, or there are exactly two pairs of adjacent bordered cyclic shifts, and the rest of the cyclic shifts of $w$ alternate between bordered and unbordered. The predicate $\isAlternatingE(e,l,m,n)$ asserts that all of the cyclic shifts of a length $n$ word at position $m$ in the Thue-Morse word alternate between unbordered and bordered, except for the $l$'th, $l+1$'th, $e$'th, and $e+1$'th cyclic shifts, which are all bordered. Note that $\isAlternatingE(n,n,m,n)$ asserts that all of the cyclic shifts of a length $n$ word at position $m$ in the Thue-Morse word alternate between unbordered and bordered.
\begin{align*}
\isAlternatingE(e,l,m,n) &:= (\forall i\, (((i \neq l \land i\neq e \land i<n-1) \Rightarrow (\isBordered(i,m,n) \Leftrightarrow \\
&\lnot \isBordered(i+1,m,n)))) \land (((i \neq l) \land (i\neq  e) \land (i=n-1)) \Rightarrow  \\
& (\isBordered(n-1,m,n) \Leftrightarrow  \lnot \isBordered(0,m,n)))) 
\end{align*}

Now we create a predicate $\hasMNUCE(m,n)$ that asserts that a length $n$ word at position $m$ in the Thue-Morse word achieves the maximum number of unbordered conjugates. 
\begin{align*}
 \hasMNUCE(m,n) &:= (\exists i,j\, ((i < j) \land (i<n-1 \land \isBordered(i,m,n) \land \isBordered(i+1,m,n)) \land  \\
 &((j=n-1 \land \isBordered(n-1,m,n) \land \isBordered(0,m,n)) \lor ((j<n-1) \land \\
 &\isBordered(j,m,n) \land \isBordered(j+1,m,n))) \land \isAlternatingE(i,j,m,n))) \lor \nonumber \\
 &\isAlternatingE(n,n,m,n) .
\end{align*}

With these predicates we can write a predicate asserting that the Thue-Morse word contains factors of every length $n>3$ that are maximally unbordered and occur at position $\leq n$. We split the computation into cases, one for even length words, and one for odd:
\begin{align*}
    \forall n \, ((n \geq 2) & \implies ( \exists i \, \hasMNUCE(i,2n)) \land i \leq 2n) \\
\forall n \, ((n\geq 2) &\implies (\exists i \hasMNUCO(i,2n + 1)) \land i \leq 2n + 1),
\end{align*}
and {\tt Walnut} evaluates these predicates to be true.
\end{proof}
Thus we have that  \[\mnuc_2(n) = \begin{cases} 
      1,    & \text{if } n=1; \\
      2,    & \text{if } n=2 \text{ or }n=3;\\
      n/2, & \text{if } n\in \{2^i,3\cdot 2^i: i\geq 1\}; \\
      n/2 - 1, & \text{if } n> 3\text{ even and }n \not\in \{2^i, 3\cdot 2^i: i\geq 1\};  \\
       \lfloor n/2 \rfloor, & \text{if } n>3 \text{ odd}.
   \end{cases}
\]

\begin{theorem}
Let $f(n) = \mnuc_2(n) - \lfloor n/2 \rfloor$.  Then $f$ is a $2$-automatic sequence.
\end{theorem}

\section{More about unbordered conjugates}

In this section we show that there exist binary words of length $n$ that have exactly $i$ unbordered conjugates where $1< i\leq \mnuc_2(n)$. 

The general idea behind the proof is to pick some $i>1$ and then pick a word $w$ of odd length such that $\nuc(w) = i$ and $\mnuc_2(|w|) = i$. Furthermore we only consider such words $w$ such that one of $w$'s conjugates contain $000$ as a factor. Then we keep adding $0$'s to $w$ precisely where $000$ first occurs. This keeps the number of unbordered conjugates the same. Then we can keep increasing the size of $w$ in this way until we hit the length we want.

\begin{lemma}\label{lemma:oddOverlap}
For $n>4$ odd, there exists a word $w\in \Sigma_2^n$ such that $\nuc(w)=\lfloor n/2\rfloor$ and $000$ is a factor of some conjugate of $w$.
\end{lemma}
\begin{proof}
By Theorem~\ref{theorem:maxUn}, such a word $w$ exists as a factor of the Thue-Morse word. It is well known that the Thue-Morse word is overlap-free. So $000$ cannot be a factor of such a word $w$. But it is possible that $w= 0u00$, or $w= 00u0$ for some word $u$. We can check whether this is the case for all odd $n>4$ by modifying our predicate from the proof of Theorem~\ref{theorem:maxUn}:
\begin{multline*}
\forall n \, ((n \geq 2) \implies 
( \exists i \, \hasMNUCO(i,2n + 1)) \land 
 ((T[i]=0 \land T[i+1]=0 \land T[2n+i]=0) \\ \lor 
(T[i]=0 \land T[2n-1+i]=0 \land T[2n+i]=0))),
\end{multline*}
which evaluates to true.
\end{proof}

\begin{lemma}\label{lemma:one000}
Let $n>4$ be odd and $w$ be a binary word of length $n$ such that a conjugate of $w$ has $000$ as a factor and $\nuc(w) = \lfloor n/2\rfloor$. Then every conjugate of $w$ contains at most one distinct occurrence of $000$ as a factor.
\end{lemma}
\begin{proof}
Suppose, contrary to what we want to prove that a conjugate of $w$ contains at least two distinct occurrences of $000$ as a factor. Call this conjugate $w'$. 

If the two occurrences of $000$ overlap, then we can write $w' = s0000t $ for some words $s,t$. Then the cyclic shifts $0ts000$, $00ts00$, and $0ts000$ are bordered. This means that only $\lfloor |ts| /2 \rfloor+1$ of the remaining cyclic shifts of $w$ can be unbordered since any unbordered cyclic shift must be followed by a bordered one. But $\lfloor |ts| /2 \rfloor +1 = \lfloor(n-4)/2\rfloor+1 < \lfloor n/2 \rfloor$, so the two occurrences of $000$ cannot overlap. 

If the two occurrences of $000$ do not overlap, then we can write $w' = s000t000$ for some words $s,t$ where $s$, and $t$ are non-empty. Then the conjugates $00t000s0$, $0t000s00$, $00s000t0$, and $0s000t00$ are bordered. By the same argument as above, of the remaining cyclic shifts, a maximum of $\lfloor |st|/2 \rfloor + 2$ of them can be unbordered. But $\lfloor |st|/2\rfloor + 2 = \lfloor (n-6)/2)\rfloor + 2 < \lfloor n/2 \rfloor$, a contradiction.
\end{proof}

\begin{lemma}\label{lemma:add0s}
Let $n>4$ be odd and $w$ be a binary word of length $n$ such that a conjugate $w'$ of $w$ has $000$ as a prefix and $\nuc(w) = \lfloor n/2\rfloor$. Then $\nuc(w)=\nuc(w') = \nuc(0^{i}w')$ for all $i\geq 0$.
\end{lemma}
\begin{proof}
Let $i\geq 0$ be an integer. We can write $w' = 000u$ for some word $u$. It is clear that $0^ju0^{i+3-j}$ is bordered for all $1\leq j \leq i+2$. Therefore, it suffices to prove that $s 000 t$ is bordered if and only if $s0^{i+3} t$ is bordered where $u = ts$.

First we prove the forward direction. Suppose $s000t$ is bordered. By Lemma~\ref{lemma:one000} we have that $s000t$ contains only one occurrence of $000$ as a factor. So $000$ is neither a prefix of $s00$ nor a suffix of $00t$. Thus, any border of $s000t$ must of length $\leq \min\{|s|,|t|\} +2$. But such a border would also be a border of $s0^{i+3} t$. 

A similar argument works for the reverse direction. Therefore $\nuc(w) = \nuc(w') = \nuc(0^iw')$ for all $i\geq 0$.
\end{proof}

\begin{theorem}
For all $1 < i \leq \mnuc_k (n)$ there exists
$w \in \Sigma_k^n$ such that
$\nuc(w) = i$.
\end{theorem}
\begin{proof}
Let $C=\{5,7,9,10,14,17\}$. For $k\geq 4$, Harju and Nowotka \cite{Harju&Nowotka:2008} showed that for all integers $i$ with $1 < i \leq n$ there exists a word $w\in \Sigma_k^n$ such that $\nuc(w) = i$. For $k=3$, Harju and Nowotka \cite{Harju&Nowotka:2008} showed that if $n\not\in C$ then for all integers $i$ with $1 < i \leq n$ there exists a word $w\in \Sigma_k^n$ such that $\nuc(w) = i$, and if $n\in C$ then for all integers $i$ with $1 < i < n$ there exists a word $w\in \Sigma_k^n$ such that $\nuc(w) = i$.

To the best of the authors' knowledge, there is no known proof of the existence of such words for $k=2$. Suppose $k=2$. By Theorem~\ref{theorem:maxUn} there exists a $w\in \Sigma_2^n$ such that $w$ is a factor of the Thue-Morse word and $\mnuc_2(n) = \nuc(w)$. So assume $i< \mnuc_2(n)$. By Lemma~\ref{lemma:oddOverlap} there exists a binary word $u$ of odd length $m$ such that $\nuc(u) = i=\lfloor m/2\rfloor$ and $000$ is a factor of some conjugate of $u$. Let $u'$ be the conjugate of $u$ such that $000$ is a prefix of $u'$. Lemma~\ref{lemma:add0s} tells us $\nuc(u)=\nuc(u') = \nuc(0^{n-m}u')$. Since $\nuc(0^{n-m}u') = i$ and $|0^{n-m} u'|=n $, we have that for all $1<i\leq \mnuc_2(n)$, there exists a $w\in \Sigma_2^n$ such that $\nuc(w) = i$.
\end{proof}

\section{Expected number of unbordered conjugates}

What is the expected number of unbordered conjugates of 
a randomly-chosen word of length $n$?

Let
$u_k(n)$ denote the number of unbordered binary words of length $n$.
It is known that $\lim_{n \rightarrow \infty} u_k(n)/k^n = \alpha_k$,
where $\alpha_k$ is a certain real number.  We have $\alpha_2 \doteq
0.2677$; see \cite{Nielsen:1973}.

\begin{theorem}
The expected number of unbordered conjugates of a randomly-chosen word of length $n$ over a $k$-letter alphabet is 
$(u_k (n) /k^n) n$, which is asymptotically equal to $\alpha_k n$. 
\end{theorem}

\begin{proof}
To see this, let $x$ be a word of length $n$.  There are
two cases:  $x$ is not primitive, and $x$ is primitive.

If $x$ is not primitive, it equals $y^e$ for some $e \geq 2$.
Then by a well-known result, every conjugate of $x$
is bordered.   So $x$ has $0$ unbordered conjugates.

If $x$ is primitive, then all its cyclic shifts are distinct.
(For otherwise $x = uv = vu$, and then by Lyndon-Sch\"utzenberger
we know $x$ is a power.)  There are $n$ of these cyclic shifts.
So if we consider all conjugates of all primitive words, each
primitive word appears exactly $n$ times.

Putting this all together, we get
\begin{align*}
\sum_{i \in S} i \cdot \Pr[X=i] &= {1 \over {k^n}} \sum_{x \in \sigma^n}
\nuc(x) \\
&= {1 \over {k^n}} \sum_{{{x \in \Sigma^n}} \atop {x \text{ primitive}}}
\nuc(x) \\
&= {1 \over {k^n}} \sum_{{{x \in \Sigma^n}} \atop {x \text{ primitive}}}
	\sum_{0 \leq i < n} [ \sigma^i (x) \text{ is unbordered } ] \\
&= {1 \over {k^n}} \cdot n \cdot \sum_{{{x \in \Sigma^n}} \atop {x \text{ primitive}}}
         [x \text{ is unbordered } ] \\
&= {1 \over {k^n}} n \cdot u_k (n), 
\end{align*}
as claimed.   Here $\nuc(x)$ is the number of unbordered conjugates
of $x$, and $\sigma^i(x)$ denotes the left shift by $i$ locations of
$x$, and $[ p ]$ is $1$ if $p$ is true and $0$ otherwise.
\end{proof}

\section{Conclusions}
We want to emphasize that our experience shows that rephrasing problems in combinatorics on words using the first-order logical theory of automatic sequences can be a useful tool in solving these problems.
We encourage others to adopt this approach.

\section*{Acknowledgments}

We thank Dirk Nowotka for helpful discussions.

\newcommand{\noopsort}[1]{} \newcommand{\singleletter}[1]{#1}

\end{document}